\documentclass{llncs}
\usepackage{amssymb}
\usepackage{graphicx}
\usepackage{subfigure}
\usepackage{enumerate}
\usepackage{amsmath}
\usepackage{epsfig}
\usepackage{xcolor}
\usepackage{arash}

\begin{document}

\title{Priority Queues with Multiple Time Fingers}
\author{Amr Elmasry\thanks{Supported by an Alexander von Humboldt Fellowship.} 
\and Arash Farzan
\and
John Iacono\inst{3}\thanks{Research partially supported by NSF grants CCF-0430849,  CCF-1018370, and an Alfred P.~Sloan fellowship, and by MADALGO---Center for Massive Data Algorithmics, a Center of the Danish National Research Foundation, Aarhus University.}
}
\institute{Max-Planck-Institut f\"ur Informatik, \\ Saarbr\"ucken, Germany
\and
Polytechnic Institute of New York Univerity\\
Brooklyn, New York, USA}

\date{}

\maketitle 

\begin{abstract}
A priority queue is presented that supports the operations insert and find-min in worst-case constant time, and delete and delete-min on element $x$ in worst-case $O(\lg(\min\{w_x, q_x\}+2))$ time, where $w_x$ (respectively $q_x$) is the number of elements inserted after $x$ (respectively before $x$) and are still present at the time of the deletion of $x$. 
Our priority queue then has both the working-set and the queueish properties, and more strongly it satisfies these properties in the worst-case sense. 
We also define a new distribution-sensitive property---the time-finger property, which encapsulates and generalizes both the working-set and queueish properties, and present a priority queue that satisfies this property.

In addition, we prove a strong implication that the working-set property is equivalent to the unified bound 
(which is the minimum per operation among the static finger, static optimality, and the working-set bounds).
This latter result is of tremendous interest by itself as it had gone unnoticed since the introduction of such bounds by  Sleater and Tarjan~\cite{splay-tree}.
Accordingly, our priority queue satisfies other distribution-sensitive properties as the static finger, static optimality, and the unified bound.

{\bf Keywords:} Data structures, splay trees, priority queues.

\end{abstract}

\section{Introduction\label{intro}}

Distribution-sensitive data structures are those for which the time bounds to perform operations vary depending on the 
sequence of operations performed~\cite{iacono-thesis}. These data structures typically perform as well as their distribution-insensitive counterparts 
on a random sequence of operations in an amortized sense; yet where the sequence of operations follows a particular distribution, or there is temporal or spatial locality in the sequence of operations, the distribution-sensitive data structures perform significantly better. 

The quintessential distribution-sensitive data structure is splay trees~\cite{splay-tree}. Splay trees seem to perform 
very efficiently (much faster than $O(\lg n)$ search time on a set of $n$ elements)  over the sequence of operations that arise naturally. 
There still exists no single comprehensive distribution-sensitive analysis of splay trees; instead, there are many theorems and conjectures that characterize the distribution-sensitive properties of splay trees. These properties include the static-optimality, static-finger bound, working-set bound, sequential-access bound,  unified bound, dynamic-finger bound, and the unified-conjecture~\cite{cole-paper,splay-tree}. 
As defined in~\cite{splay-tree}, the ``unified-bound" is the per-operation minimum of the static-optimality, static-finger, and working-set bounds. 
By the ``unified-conjecture'', we assume the definition given in~\cite{cole-paper}, which subsumes both the dynamic-finger and working-set bounds. 
We refer the reader to~\cite{cole-paper} and also~\cite{splay-tree} for a thorough definition and discussion of these properties. 

Some of these properties imply others. Figure~\ref{implication-diagram} illustrates such relationship between these properties.
As a new contribution, we prove in Section~\ref{unified-from-working-set} the implication of the unified bound from the working-set bound. We also show that the working-set bound on two sequences is asymptotically the same as the working set bound of any interleaving of these sequences; the full proof appears in Appendix~\ref{sec:mws}.

\begin{figure}[t]
	\centering\includegraphics[scale=0.4]{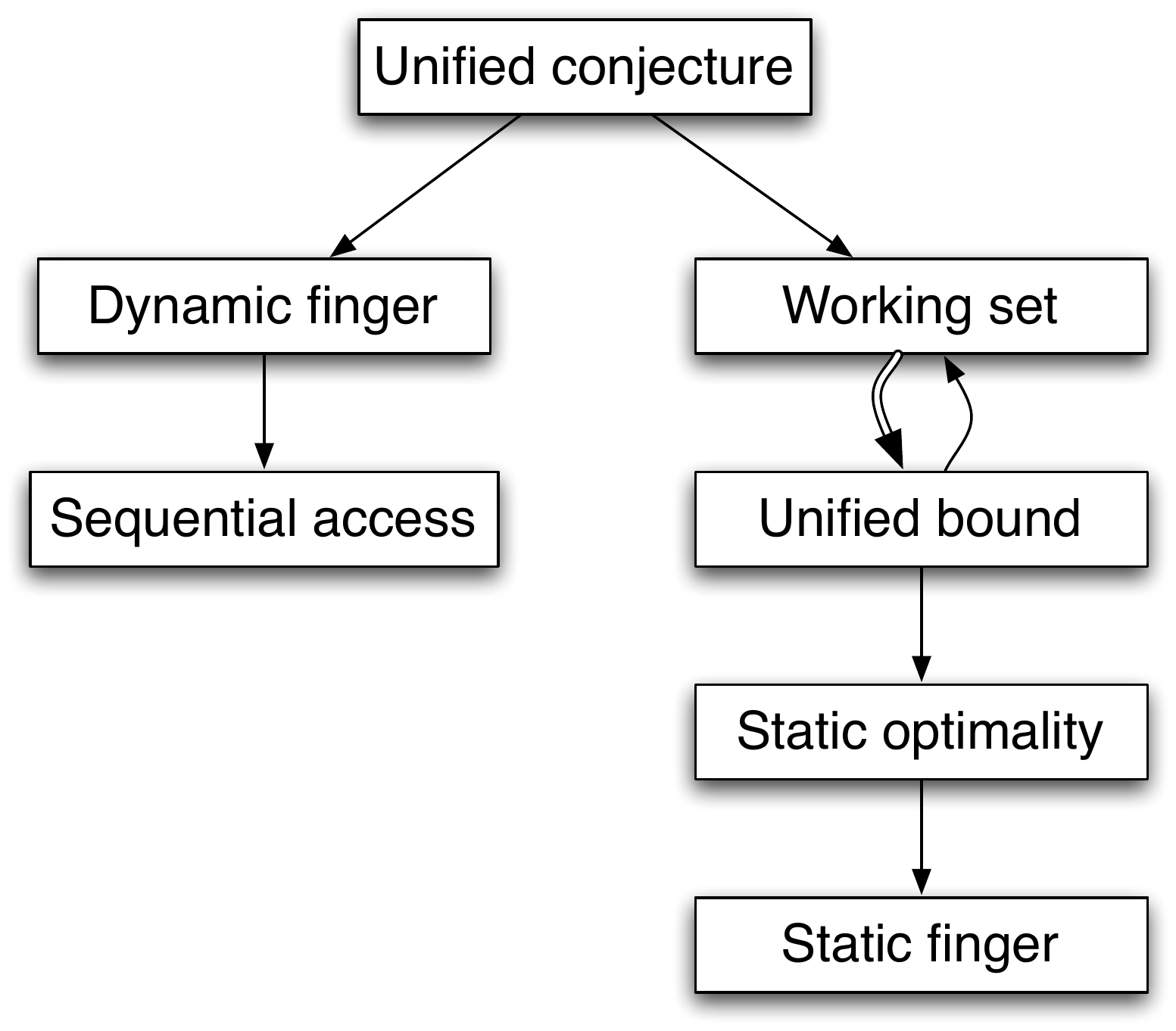}
	\caption{The implication relationships between various distribution-sensitive properties. The implication of the unified bound from 
	the working-set property is a contribution of Section~\ref{unified-from-working-set} of this paper.\label{implication-diagram}}
\end{figure}

Distribution-sensitive data structures are not limited to search trees. Priority queues have 
also been designed and analyzed in the context of distribution-sensitivity~\cite{funnel-heap,elmasry,iacono-pairing-heaps,iacono-thesis,queaps}.
It is easy to observe that a priority queue with constant insertion time cannot have the sequential-access property (and hence cannot as well have the dynamic-finger property), 
for otherwise a sequence of insertions followed by a sequence of minimum-deletions give the elements in sorted order in linear time. Therefore, the working-set property has 
been of main interest for priority queues. Informally, the working-set property states that elements that have been recently updated are faster to update again compared to the elements that have not been accessed in the recent past. Iacono~\cite{iacono-pairing-heaps} proved that pairing-heaps~\cite{pairing-heaps} satisfy the working-set property as follows;
in a heap of maximum size $n$, it takes $\oh{\lg\min\set{n_x,n}}$ amortized time to delete the minimum element $x$, where $n_x$ is the number of operations performed since $x$'s insertion. 
Funnel-heaps are I/O-efficient heaps for which it takes $\oh{\lg\min\set{i_x + 2,n}}$ to delete the minimum element $x$, where $i_x$ is the number of insertions made since $x$'s insertion (we note that $i_x \le n_x$). Elmasry~\cite{elmasry} gave a priority queue supporting delete-min in $\oh{\lg(w_x+2)}$ worst-case time, where $w_x$ is the number of elements inserted after $x$ and are still present in the priority queue ($w_x \le i_x \le n_x$). 

None of these results supports delete with the working-set bound. In section~\ref{delete-section}, we present a priority queue that supports 
{\em both} delete and delete-min in $\oh{\lg(w_x+2)}$ worst-case time (and insertion in worst-case constant time).

One natural sequence of  operations  in a data structure is a first-in first-out type of updates. Data structures sensitive to 
these sequences must operate fast on elements that have been least recently accessed. This distribution-sensitive property is 
referred to as the ``queueish'' property in~\cite{queaps}. In the context of priority queues, such property states that 
the time to perform delete or delete-min on an element $x$ is $\oh{\lg\paren{q_x+2}}$, , where $q_x$ is the number of elements inserted before $x$ and are still present in the priority queue.
Note that $q_x = n- w_x$, where $n$ is the number of elements currently present in the priority queue.
However, it is shown in~\cite{queaps} that no binary search tree can be sensitive to this property. 
Albeit, a priority queue with the queueish property is presented in the same paper.

It remained open whether there exists a priority queue sensitive to both the working-set  and the queueish properties. 
We resolve the question affirmatively by presenting such a priority queue in Section~\ref{queueish-section}. This priority queue 
is the most comprehensive distribution-sensitive priority queue to date.

In Section~\ref{fingers-section}, we present a more powerful priority queue that incorporates multiple  time fingers.
We define time fingers $t_1, t_2, \ldots, t_c$ as points of time during the sequence of updates, which are set on-line as they arrive.
We define the working-set of an element $x$ with respect to time finger $t_i$, $w_x(t_i)$, as the number of elements that have been 
inserted in the window of time between the insertion of $x$ and $t_i$ and are still present in the priority queue.
We say a priority queue satisfies the multiple time finger property if the time to delete or delete-min $x$ is $\oh{\lg(\Min_{i=1}^c\set{w_x(t_i)}+2)}$. 
It is not hard to see that the working-set property is equivalent to having a single time finger of 
$t_1=+\infty$, and the queueish property is equivalent to having a single time finger of $t_1=0$. 
The priority queue presented in Section~\ref{queueish-section}, which supports both the working-set and queueish properties, is equivalent to having two time fingers $t_1=0, t_2=+\infty$. 
In Section~\ref{fingers-section}, we present a priority queue that satisfies the property for a constant number of time fingers.




\section{From the working-set bound to the unified bound \label{unified-from-working-set}}

The {\em static finger} property states that, for any fixed element $f$ (the finger), the amortized time to access an element $a$ is $\oh{\lg(d(a,f)+2)}$, where 
$d(a,f)$ is the rank difference between $a$ and $f$. More specifically, for a sufficiently long sequence of accesses $x_1, x_2, \ldots, x_m$, the 
total access time is $$\oh{\Sum_{i=1}^m \lg(d(x_i,f)+2)}.$$ 

The static optimality property (entropy bound) states that, for a sequence of accesses $x_1, x_2 \ldots, x_m$, where the element corresponding to $x_i$ is accessed $q(x_i)$ times in the entire sequence, the total access time is $$\oh{\Sum_{i=1}^m \lg\paren{\frac{m}{q(x_i)}+1}}.$$

The working set $w_X(i)$ of an operation $x_i$ in sequence $X$, which accesses element $x_i$, is defined as the number of distinct items accessed since the last access to $x_i$. The working-set property states that, for a sufficiently long sequence of accesses $x_1, x_2 \ldots, x_m$, the total access time is $$\oh{\Sum_{i=1}^m\lg(w_X(i)+2)}.$$

We observe that the working-set bound of two sequences does not asymptotically change when those two sequences are arbitrarily interleaved. We state this theorem formally below and prove it in Appendix~\ref{sec:mws}. This result will be needed to prove the main claim of this section, and is interesting in its own right.

\begin{theorem} \label{th:merge}
Let $X$ be a search sequence and 
let $Y$ and $Z$ be two subsequences of $X$ that partition $X$. Stated another way, $X$ is an interleaving of $Y$ and $Z$. Then, 

$$\sum_{i=1}^{|X|} \lg (w_X(i)+2 ) = \Theta\left( 
\sum_{i=1}^{|Y|} \lg (w_Y(i)+2 )+
\sum_{i=1}^{|Z|} \lg (w_Z(i)+2 )
\right).$$
\end{theorem}

Iacono~\cite{iacono-thesis} observed that the working-set property implies the static-optimality and static-finger properties. Therefore, the working-set property is the strongest of the three properties. The unified bound indicates an apparently stronger property than the three properties.
The unified bound states that the total time for a sufficiently long sequence of accesses $X=x_1, x_2 \ldots, x_m$ and any fixed finger $f$ is
\begin{equation} \label{eq:unified} \oh{\Sum_{i=1}^m \lg\min\set{d(x_i,f)+2, \, \frac{m}{q(x_i)} + 1, \, w(i)+2 }}.
\end{equation}

However, we show next that the working-set property is asymptotically as strong as the unified property.

\begin{theorem}
The working-set bound is asymptotically equivalent to the unified bound.
\end{theorem}

\begin{proof}
Clearly, the unified bound implies the working-set bound as $$\Sum_{i=1}^m \lg\min\set{d(x_i,f)+2, \, \frac{m}{q(x_i)} + 1, \, w(i)+2 } \,\le \, \Sum_{i=1}^m\lg(w(i)+2).$$

Let $Y$ be the subsequence of $X=x_1, x_2, \ldots x_m$ consisting of those elements $x_i$ where
$w(i)+2  = \min\set{d(x_i,f)+2, \, \frac{m}{q(x_i)} + 1, \, w(i)+2 }$, and let $\overline{Y}$ be the subsequence of $X$ crated by removing $Y$. Let $m(i)=j$ if $y_i$ is $x_j$, and let $\overline{m}(i)=j$ if $\overline{y}_i=x_i$. We also subscript the $w(\cdot)$, $d(\cdot,\cdot)$ and $q(\cdot)$ to explicitly indicate which sequence these measures are with respect to. Then:

\begin{align}
\label{eq2}
&\Sum_{i=1}^m \lg\min\set{d_X(x_i,f)+2, \, \frac{m}{q_X(x_i)} + 1, \, w_X(i)+2 } 
\\ 
\label{eq3}
 &= \Sum_{i=1}^{|Y|} \lg (w_Y({m(i)})+2) 
+\Sum_{i=1}^{|\overline{Y}|} \lg\min\set{d(x_{\overline{m}(i)},f)+2, \, \frac{m}{q(x_{\overline{m}(i)})} + 1 }  
\\
\label{eq4}
 & = \Omega\left( \Sum_{i=1}^{|Y|} \lg (w_Y({m(i)})+2) 
+ \Sum_{i=1}^{|\overline{Y}|} \lg  \frac{\oh{1}}{\max \set{ \frac{q(x_{{m}(i)})}{m},\frac{1}{(d(x_{{m}(i)},f)+2)^2}  } } \right)
\\
\label{eq5}
&= \Sum_{i=1}^{|Y|} \lg (w_Y({m(i)})+2) 
+\Omega \left( \Sum_{i=1}^{|\overline{Y}|} \lg  \left( \frac{m}{q_{\overline{Y}}(x_{\overline{m}(i)})} + 1 \right)  \right) 
\\
\label{eq6}
&= \Sum_{i=1}^{|Y|} \lg (w_Y({m(i)})+2) 
+\Omega \left( \Sum_{i=1}^{|\overline{Y}|} \lg (w_{\overline{Y}}({\overline{m}(i)})+2)  \right) 
\\
\label{eq7}
 &= \Omega \left( \sum_{i=1}^{m} \lg (w_X(i)+2)  \right)
\end{align}

Equation~\ref{eq3} splits one sum into two using the partitioning of $X$ into $Y$ and $\overline{Y}$. The left term of Equation~\ref{eq4} is obtained by replacing $w_X$ with $w_Y$, which only can cause a decrease. The right sum of Equation~\ref{eq4} is a static-optimality type formula, which is asymptotically the same as the right sum of Equation~\ref{eq3}. This has the property that if $y_i=y_j$, the $i$th and $j$th term of this sum are identical, and we call this value the \emph{weight} of an element. The $O(1)$ numerator is chosen so that the weights sum to 1. Since this is just a static-optimality type weighting scheme, this sum is at least the entropy of the frequencies, and this observation leads to Equation~\ref{eq5}. To get from Equation~\ref{eq5} to Equation~\ref{eq6} the fact that the static-optimality bound is big-Omega of the working-set bound; this is Theorem~10 of~\cite{iacono-thesis}. Moving from Equation~\ref{eq6} to Equation~\ref{eq7} requires the observation that the sum of the working-set bounds of two sequences of operations are asymptotically the same as the working-set bound of the interleaving of two sequences; this is stated above and proved in Appendix~\ref{sec:mws} as Theorem~\ref{th:mergeappendix}.

\end{proof}

\section{A priority queue with the working-set property\label{review-section}}

Our priority queue builds on the priority queue in~\cite{elmasry}, which supports insertion in  constant time while the minimum deletion fulfills the 
working-set bound. The advantage of the priority queue in~\cite{elmasry} over those in~\cite{funnel-heap,pairing-heaps,iacono-pairing-heaps} 
is that it satisfies the stronger working-set property in which elements that are deleted do not count towards the working sets.
Next, we outline the structure of this priority queue.

The priority queue in~\cite{elmasry} comprises heap-ordered $(2,3)$ binomial trees. As defined in~\cite{elmasry}, the subtrees of the root of a $(2,3)$ binomial tree of rank $r$ are $(2,3)$ binomial trees; there are one or two children having ranks $0, 1, \dots, r-1$, ordered non-decreasing from right to left. It is trivial to verify that the rank $r$ of an $n$-node $(2,3)$ binomial tree is $\thetah{\lg n}$. Figure~\ref{fig:binomial} illustrates the recursive structure of a $(2,3)$ binomial tree.

\begin{figure}[t]
\centering\includegraphics[scale=0.4]{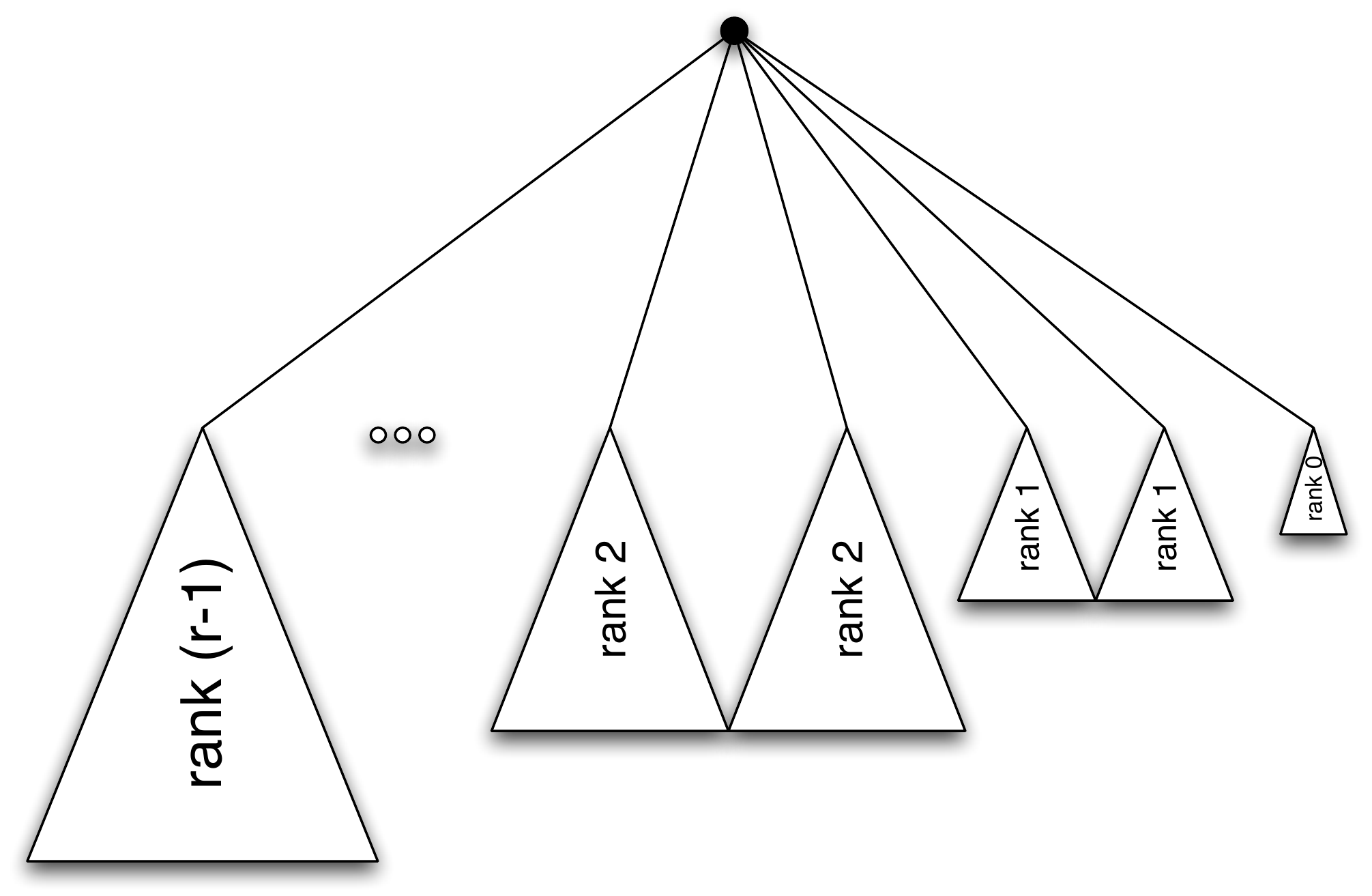}
\caption{The recursive structure of a $(2,3)$ binomial tree of rank $r$: Subtrees rooted at the children of the root are $(2,3)$ binomial trees of ranks $0,1,\ldots, r-1$. The sequence of ranks 
forms a non-decreasing sequence from right to left such that each value from $0, 1,\ldots, r-1$ occurs either once or twice.\label{fig:binomial}} 
\end{figure}

The ranks of the $(2,3)$ binomial trees of the priority queue are as well non-decreasing from right to left. 
For the amortized solution, there are at most two trees per rank. The main obstacle against achieving the bounds in the worst case is the possibility that a long 
sequence of consecutive ranks would have no corresponding trees. To overcome this problem, in the worst-case solution, the number of trees per rank obey an extended-regular number system
that imposes stronger regularity constraints, which implies that the ranks of any two adjacent trees differ by at most $2$ (see~\cite{elmasry} for the details). 

The root of every $(2,3)$ binomial tree has a pointer to the root with the minimum value among those to the left of it. Such {\em prefix-minimum} pointers allow for finding the overall minimum 
element in constant time, with the ability to maintain such pointers after deleting the minimum in time proportional to the rank of the deleted node. 
Figure~\ref{fig:oneSided} illustrates how various $(2,3)$ binomial trees constitute the priority queue.
\begin{figure}[t]
\centering\includegraphics[scale=0.5]{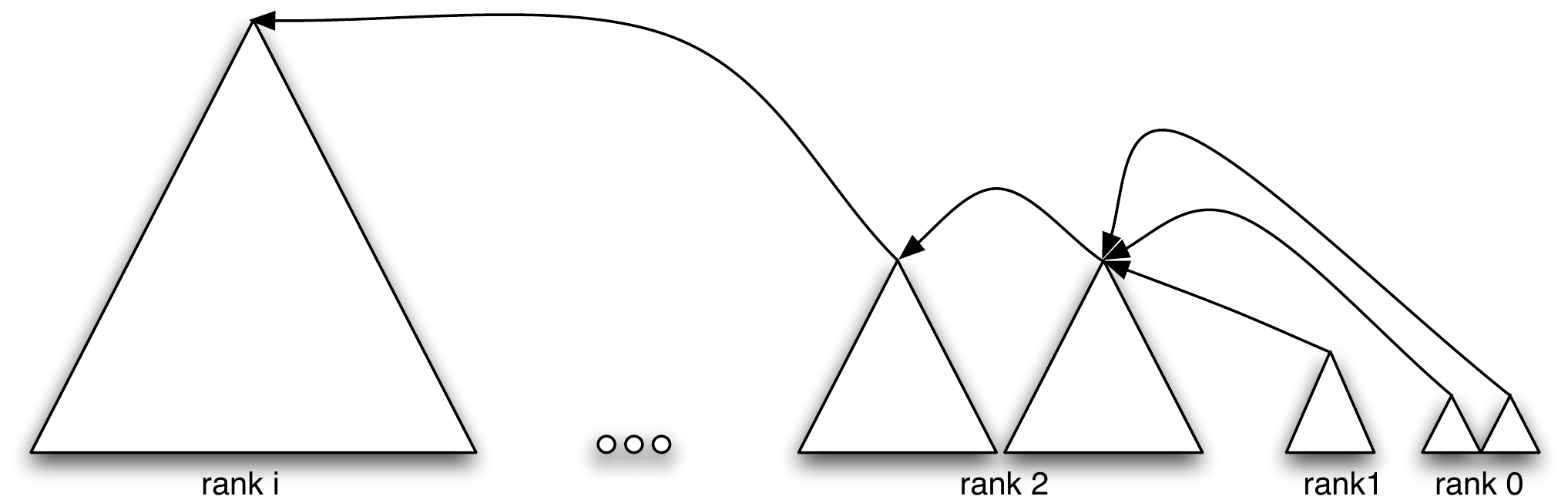}
\caption{$(2,3)$ binomial trees comprise the priority queue: The rank of the trees are non-decreasing from right to left. This figure illustrates the amortized solution where there are at most two trees per rank. In the worst-case solution, the number of trees of each rank follows a much stricter number system. 
Prefix-minimum pointers are maintained at the root of the trees. Each tree root points to the minimum root to the left of it.\label{fig:oneSided}}
\end{figure}

A total order is maintained indicating the time the elements were inserted. We impose that across binomial trees, if binomial tree $T_1$ is to the right of another $T_2$, then all elements in $T_1$ have been inserted after those in $T_2$. Furthermore within an individual binomial tree, the preorder ordering of elements with a right-to-left precedence to subtrees must be chronologically consistent with the insertion time of these elements. When performing operations, we occasionally disobey this ordering by reversing the order of two entire subtrees. We mark these points by maintaining a {\it reverse bit} with every node $x$; such reverse bit indicates whether the elements in $x$'s subtree were inserted before or after the elements in $x$'s parent and those in the descendants of the right siblings of $x$. 

Two primitive operations are {\it split} and {\it join}. A tree of rank $r$ is split to two or three trees of rank $r-1$; this is done by detaching the one or two children of the root having rank $r-1$.
On the other hand, two or three trees of rank $r-1$ can be joined to form a tree of rank $r$; this is done by making the root(s) with the larger value the leftmost child(ren) of the other, and setting the reverse bit(s) correctly. To join a tree of rank $r-1$ and a tree of rank $r-2$, we split the first tree then join all the resulting trees; the outcome is a tree which has rank either $r-1$ or $r$.  
With these operations in hand, it is possible to detach the root of a $(2,3)$ binomial tree of rank $r$ and reconstruct the tree again as a $(2,3)$ binomial tree with rank $r-1$ or $r$; this is done by repeated joins and splits starting from the rightmost subtrees of the deleted root to the leftmost (see~\cite{elmasry} for the details). 

To {\em insert} an element, a new single node is added as the rightmost tree in the priority queue. This may give rise to several links once there are three trees with the same rank; the number of such links is amortized as a constant, resulting in the constant amortized cost per insertion. After every link, the prefix-minimum pointer of the surviving root may need to be updated. For the worst-case solution, the underlying number system guarantees at most one join per insert. 

To perform {\em delete-min}, the tree $T$ of the minimum root is identified via the prefix-minimum pointer of the rightmost root, the tree $T$ is reconstructed  as a $(2,3)$ binomial tree after detaching its root. This may be followed by a split and a join if $T$ has rank one less that its original rank. Finally, the prefix-minimum pointers are updated. For the amortized solution, several splits of $T$ may follow the delete-min operation. Starting with $T$, we repeatedly split the rightmost tree resulting from previous splits until such tree and its right neighbor (the right neighbor of $T$ before the delete-min) have consecutive ranks; this splitting is unnecessary in the worst-case solution. It is not hard to conclude that the cost of delete-min is $O(r)$, where $r$ is rank of the deleted node. In the worst-case solution, the rank of the deleted node $x$ is $O(\lg{(w_x+2)})$. For the amortized solution, an extra lemma would prove the same bound in the amortized sense.  

\section{Supporting delete within the working-set bound\label{delete-section}}

The existing distribution-sensitive priority queues~\cite{funnel-heap,elmasry,iacono-pairing-heaps,iacono-thesis,queaps} do 
not support {\em delete} within the working-set bound. In this section, we modify the priority queue outlined in section~\ref{review-section} to
support deletion within the working-set bound. 

Including delete in the repertoire of operations is not hard but should be done carefully. The major challenge  is to correctly maintain the total order imposed by the reverse bits following deletions. 

We start by traversing upwards via the parent pointers from the node $x$ to-be-deleted until the root of the tree of $x$ is reached. Then starting at this root, the current subtree is repeatedly split into two or three trees, one or two of them are pushed to a stack while continuing to split the tree that contains $x$, until we end up with a tree whose root is $x$. At this stage, we delete $x$ analogously to the delete-min operation; the node $x$ is detached and the subtrees resulting from removing $x$ are
incrementally joined from right to left, while possibly performing one split before each join (similar to the delete-min). 

We now have to work our way up to the root of the tree and merge all subtrees which we have introduced by splits on the way down from the root. The one or two trees that have the same rank are repeatedly popped from the stack and joined with the current tree, while possibly performing one split before each join (as required for performing a join operation). Once the stack is empty, a split and a join may be performed if the resulting tree has rank one less that its original rank (again analogously to the delete-min operation). 

The total order is correctly maintained by noting that the only operations employed are the split and join, which are guaranteed to set the reverse bits correctly~\cite{elmasry}. Since the height of a $(2,3)$ binomial tree is one plus its rank, the time bound for delete is $O(r)$, where $r$ is the rank of the tree that contains the deleted node. This estqablishes the same time bound as that for delete-min in both the amortized and worst-case solutions (see~\cite{elmasry} for the details).

\begin{theorem}
The priority queue presented in this section performs find-min and insert in constant time, and both delete and delete-min of an element $x$ in $\oh{\lg \paren{w_x+2}}$ time , where $w_x$ is the number of elements inserted after $x$ and are still present at the time of $x$'s deletion.
\end{theorem}

\section{Incorporating the queueish property\label{queueish-section}}

The queueish property for priority queues states that the time to perform delete or delete-min on an element $x$ is $\oh{\lg\paren{n-w_x+2}}$, where $n$ is the number of elements currently present in the priority queue, and $w_x$ is the number of elements inserted following the insertion of $x$ and are still present. In other words, the queueish property states that the time to perform delete or delete-min on an element $x$ is $\oh{\lg\paren{q_x+2}}$, where $q_x = n -w_x$ is the number of elements inserted prior to insertion of $x$ and are still present in the priority queue. 
Queaps~\cite{queaps} are queueish priority queues that support insert in amortized constant time and support delete-min of an element $x$ in amortized $\oh{\lg(q_x+2)}$ time. 

We extend our priority queue with the working-set bound to also support both delete and delete-min within the queueish bound. Accordingly, the priority queue simultaneously satisfies both the working-set and the queueish properties.
Instead of having the ranks of the trees of the queue non-decreasing from right to left, we split the queue in two sides, a right queue and a left queue, forming a double-ended priority queue. The ranks of the trees of the right queue are monotonically non-decreasing from right to left (as in the previous section), and those of the left queue  
are monotonically non-decreasing from left to right. We also impose the constraint that the difference in rank between the largest tree on each side is at most one.  Figure~\ref{fig:doubleEnded} depicts the new priority queue.
\begin{figure}[t]
\centering\includegraphics[scale=0.6]{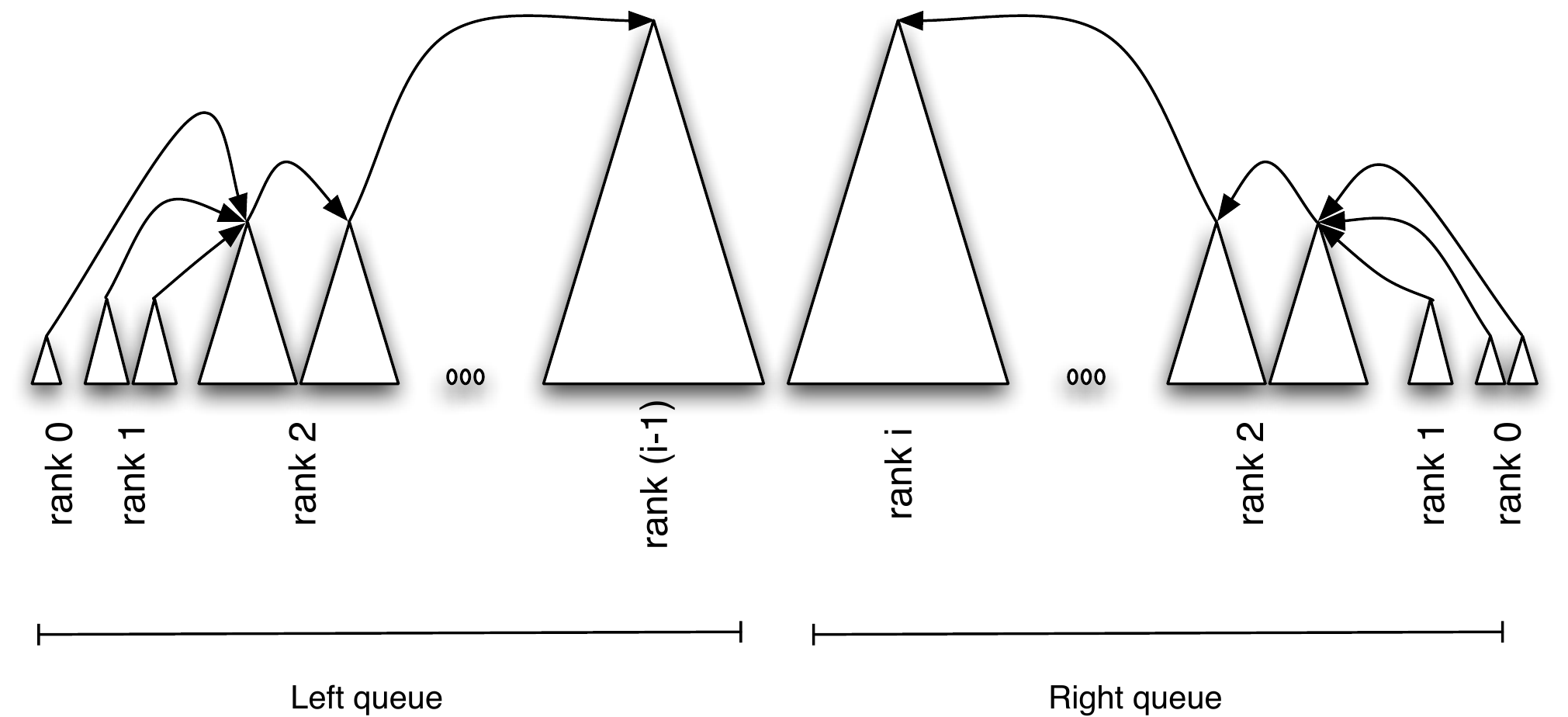}
\caption{The priority queue satisfying the queueish property: the priority queue comprises of two queues one of which has tree ranks increasing from right to left (right queue) and one increasing from left to right (left queue). We also maintain that the ranks of the largest trees in the queues must differ by at most one. The prefix-min pointers 
are stored independently in both sides.
\label{fig:doubleEnded}} 
\end{figure}

The prefix-minimum pointers in the left and right queues are kept independently. In the right queue, the root of each tree maintains a pointer to the root with the minimum value among those in the right queue to the left of it. Conversely, in the left queue, the root of each tree maintains a pointer to the root with the minimum value among those in the left queue to the right of it. To find the overall minimum value, both the left and right queues are probed.

Insertions are performed exactly as before in the right queue. The delete-min operation is performed in the left or right queue depending on where the minimum lies. Deletions are also performed as mentioned in the previous section. 

However, we must maintain the invariant that the difference in rank between the largest tree in the left and right sides is at most one.
Since the total order is maintained among our trees, this invariant guarantees that the rank of the tree of an element $x$ is $\oh{\lg(\min\set{w_x, q_x}+2)}$.
As a result of an insertion or a deletion, the difference in such ranks may become two. Once the largest rank on one side is two more than that on the other side, the trees with such largest rank are split each in two or three trees, and the appropriate tree among the resulting ones is moved to the other side, increasing the largest rank on the second side by one. As a result of those splits, the number of trees of the same rank on the first side may now exceed the limit, and hence a constant number of joins would be needed to satisfy the constraints. 
Once a tree is moved from one side to the other, the prefix-minimum pointers of the priority queue on the second side need to be updated. Because such action happens only after a lot (linear) number of operations, updating the prefix-minimum pointers only accounts for a constant extra in the amortized cost per operation. If we want to guarantee the costs in the worst case, updating those prefix-minimum pointers is to be done incrementally with the upcoming operations. 

A deletion of a node $x$ in a tree of rank $r$ would still cost $O(r)$ time, but now $r = O(\lg{(\min\set{w_x, q_x}}+2))$ in the amortized sense (for the amortized solution) or in the worst-case sense (for the worst-case solution).

\begin{theorem}\label{thm:double-ended}
	The priority queue presented in this section performs find-min and insert in constant time, and both delete and delete-min of an element $x$ in $\oh{\lg(\min\set{w_x, q_x}+2)}$ time, where $w_x$ and $q_x$ are the number of elements inserted after, respectively before, $x$ and are still present at the time of $x$'s deletion.
\end{theorem}
 
\section{Supporting multiple time fingers\label{fingers-section}}

In this section, we introduce a new distribution-sensitive property, which encapsulates both the working-set and the queueish properties.
We refer to this property as the {\em multiple time-fingers} property.
Time fingers $t_1, t_2, \ldots, t_c$ are points of time during the sequence of updates which are set and fixed as they arrive. In other words, 
as the time progresses with the sequence of operations at multiple occasions the user can specify the time being as a time finger. The elements inserted 
in the temporal vicinity of these time-fingers must be accessible fast.

We define, $w_x(t_i)$, the working-set of an element $x$ with respect to time finger $t_i$,  as the number of elements that have been 
inserted in the window of time between the insertion time of $x$ and time $t_i$ and are still present in the priority queue at the time of $x$'s deletion. 
We say a priority queue satisfies the multiple time-finger property if the time to perform delete or delete-min operations on an element $x$ is $O(\lg(\Min_{i=1}^c\set{w_x(t_i)}+2))$. 
Clearly, the priority queue we have presented so far corresponds to a priority queue with the time-finger property for 
two time fingers of $t_1=0$ (the queueish property) and $t_2=+\infty$ (the working-set property). In this section, 
we present a priority queue that  satisfies the time-finger property for any constant number of time fingers.

The structure consists of multiple double-ended priority queues of Section~\ref{queueish-section}. We start 
with a single copy of a double-ended priority queue $PQ_0$ at the beginning and at each point when a new time finger is 
introduced we finalize the priority queue and start a new one.  Therefore, corresponding to $c$ time-fingers $t_1=0, \ldots, t_c=\infty$, we 
have $c-1$ double-ended priority queues $PQ_1, \ldots, PQ_{c-1}$. 

Insertions are performed in the last (at the time when the insertion is performed) priority queue, and by Theorem~\ref{thm:double-ended} take constant time. For delete operations, 
we are given a reference to an element $x$ to delete, we determine to which priority queue $PQ_j$ the element belongs and delete it. 
This requires $\oh{\lg(\min\set{w_x(t_{j}),w_x(t_{j+1})}+2)}$ time, as indicated by Theorem~\ref{thm:double-ended}. 
Since $x$ belongs to $PQ_j$, for any $i < j$, $w_x(t_j) \le w_x(t_i)$, and for any $i > j+1$, $w_x(t_{j+1}) \le w_x(t_i)$. 
It follows that $\lg(\min\set{w_x(t_j),w_x(t_{j+1})}+2) = \lg(\Min_{i=1}^c\set{w_x(t_i)}+2)$.
For delete-min operation, it suffices to note that the find-min operation takes constant time in the double-ended priority queue 
of Theorem~\ref{thm:double-ended}. Therefore, we can determine in constant time which priority queue contains the minimum and perform the delete-min operation in there. 
The running time argument is the same as that for the delete operation.

\begin{theorem}
	Given a constant number of time fingers $t_1=0, t_2, \ldots, t_c=\infty$, the priority queue presented in this section performs find-min and insert in constant time, and both delete and delete-min of an element $x$ in $O(\lg(\Min_{i=1}^c\set{w_x(t_i)}+2))$ time, where $w_x(t_i)$ is the number of elements that have been inserted in the window of time between the insertion time of $x$ and time $t_i$ and are still present at the time of $x$'s deletion.
\end{theorem}
 
\section{Conclusion and future work}

We presented a hierarchy of distribution-sensitive properties in Figure~\ref{implication-diagram}. We established that the working-set property is equivalent to the unified-bound property. The queueish property introduced by~\cite{queaps} is missing from the picture as it is neither derived by or implies  any other property. Nevertheless, we argued that it is a very natural distribution-sensitive property.

We considered the case of distribution-sensitive priority queues. 
Precisely speaking, we designed a priority queue that supports insertions in constant time and delete-min and delete operations in distribution-sensitive bounds.
Provably, priority queues cannot satisfy the sequential-access property and in accordance neither the dynamic-finger nor the unified properties. We therefore focused on other distribution-sensitive properties, namely: the working-set and the queueish properties. We presented a priority queue that satisfies both properties.
Our priority queue build on the priority queue of~\cite{elmasry}, which supports insertion in constant time and delete-min in the working-set time bound. We showed that the same structure can also support delete operations within the working-set bound. We then modified the structure to satisfy the queueish property as well. 
It is worthy to note that the priority queue designed supports the stronger definition of the working-set and the queueish properties in which the elements deleted do not influence the time bounds.

Our result about the equivalence of the working-set property and the unified-bound property then implies that our priority queue also satisfies the unified-bound, static-optimality and static-finger properties. 

We defined the notion of time fingers, which encapsulate the working-set and the queueish properties. The priority queue described thus far corresponds to a priority queue that supports two time fingers. We extended the support to any constant number of time fingers. 

The bounds mentioned are amortized. However, we showed that the time bounds for the working-set and queueish properties can also be made to work in the worst case. More generally, the multiple time-finger bounds can be made to work in the worst case. However, the time bounds for other properties: unified bound, static optimality, and static finger naturally remain amortized.  

As for future work, one key operation is still missing from the supported repertoire of operations; that is the {\em decrease-key} operation. We leave open the question of whether decrease-key operations can also be performed in constant time while supporting the distribution-sensitive bounds for delete operations.

\bibliographystyle{plain}
\bibliography{refs}

\vfill \pagebreak

\appendix

\section{The effect of merging sequences on the working-set bound}
\label{sec:mws}

\begin{theorem} \label{th:mergeappendix}
Let $X$ be a search sequence and 
let $Y$ and $Z$ be two subsequences of $X$ that partition $X$. Stated another way, $X$ is an interleaving of $Y$ and $Z$. Then, 

$$\sum_{i=1}^{|X|} \lg (w_X(i)+2 ) = \Theta\left( 
\sum_{i=1}^{|Y|} \lg (w_Y(i)+2 )+
\sum_{i=1}^{|Z|} \lg (w_Z(i)+2 )
\right).$$

\end{theorem}

\begin{proof}

The $\Omega$ direction is immediate and thus we focus on the $O$ direction.

Let $m_X(i)=j$ if $y_i$ corresponds to $x_{m(i)}$ in the subsequence relationship of~$Y$ with regard to~$X$. Let $m_Y$ and all subsequent notion in this section subscripted on $Y$ have an analgous definition on $Z$.

Let $\omega_X(i)$ be the largest $j<i$ such that $x_i=x_j$. That is, $x_{\omega_X(i)}$ is the previous access to the operation $x_i$ in $X$. Let $W_X(i)$ be the set of indicides $j$ such that $x_j$ is the first occurence of that element  in the range $x_{\omega_X(i)+}..x_{i-1}$. This definition is such that $W_X(i)$ is the set envisioned by the concept of the ``working set'' of $x_i$ and is constructed so that $|W_X(i)|=w_X(i)$.

Observe that the largest $j$ in $W_X(i)$ has the property that its working set number is at least $w_X(x_i)-1$ since it is the first occurrence of the value $x_j$ where $j \in W_X(i)$, and there are $|W_X(i)|-1=w_X(x_i)-1$ different before it. In general, the $k$th largest $j$ in $W_X(i)$ has a working set number at least $w_X(x_i)-j$. Let $W'_X$ be the $|W_X(i)|/2$ largest elements of $W_X$.
This all $\left| \frac{1}{2}w_X(i) \right|$ elements $j$ of $W'_X$ have working set number at least $\frac{1}{2}w_X(i)$.

Let $A_i$ be formally defined to be:

$$ A_i= \set{ j | w_X({m_Y(j)})\geq  w_Y(j)^2 \text{ and }  i=\lfloor \lg w_X({m_Y(j)})\rfloor }  $$

That is, it consists of the indicies of those elements in $Y$ whose  logarithm of its working set number is double in $X$ relative to $Y$ and where it is in the range $[2^{2^i} .. \cdot2^{2^{i+1}})$. All sets $A_i$ are thus disjoint, and all indicies $i$ in $[1..|Y|]$ are in some set $A_i$ unless the log of the working set of the element $y_i$ does not change by more than a factor of two as a result of the merge with $Z$.

Now, pick some element $j \in A_i$. Recall that $W_X(i)$ is the working-set of $x_i$ in $X$, and $W'_X(i)$ is half of the elements of $W_X(i)$. Some of these elements come from $Y$, and some from $Z$. However, the vast majority come from $Z$, since the total number is at least the number from $Y$ squared. Very conservatively, at least half of the elements in $W'(i)$ have some $j$ such that $m_Z(j)=i$. We say that these elements of $Z$ are covered by the element $j$ at level $i$; this set is represented by $C_i(j)$ and is defined as follows, and then bounded:

$$ C_i(j) = \set{ k | m_Z(k) \in W'_X(j)}  $$

$$2^{2^{i+1}}\geq w_X(j)=| W_X(j)| \geq |C_i(j)| \geq \frac{1}{2}| W'_X(j)|=\frac{1}{4}| W_X(j)|=\frac{1}{4}w_X(j)\geq\frac{1}{4}2^{2^{i}}$$

By construction, for any fixed $k$, there are only at most $2^{2^{i+1}}$ elements $j$ such that $C_i(j)=k$.
Thus the size of the the covered set of the union of all $C_i(j)$ where $j \in A_i$ is $\Omega(|A_i|)$; we denote this set as $C_i$. So, therefore:

$$ \sum_{i=0}^{\infty} |C_i|2^i \leq  \sum_{i=1}^{|Z|} 2\log w_Z(z_i) $$

This is because for each $z_i$, contributes at most to each $|C_i|$ where $w_Z(z_i)\leq 2^{2^{i+1}}$ and is not part of any $C_i$  when $w_Z(z_i) > 2^{2^{i+1}}$.

Putting this information together gives:

\begin{align*}
\sum_{i=1}^{|Y|} w_X(x_{m(i)}) - \sum_{i=1}^{|Y|}w_Y(x_{m(i)})
&= \Theta\left( |Y|+\sum_{i=0}^\infty |A_i|2^i \right)
\\
& = O\left( |Y|+\sum_{i=0}^\infty |C_i|2^i \right)
\\
& = \oh { |Y|+\sum_{i=1}^{|Z|} \log w_Z(z_i) }
\end{align*}

Taking this equation, linearly combining it with the symmetric version where $Y$ and $Z$ are transposed, and noting that $|Y|$ and $|Z|$ are lower-order terms in the resultant equation yields the claim of this Theorem.

\end{proof}

\end{document}